\documentclass{article}
\usepackage{amsmath,amsfonts,amssymb,amsthm,graphicx}
 \usepackage{dsfont}    
\usepackage{enumerate}
\usepackage{color}
\usepackage{natbib}
\usepackage{multirow}
\usepackage{comment}

\usepackage{booktabs,subcaption,dcolumn}

\usepackage{tikz}
\usepackage{caption}
\usetikzlibrary{decorations.pathreplacing}

\usetikzlibrary{cd} 
\usepackage{tkz-graph}
\usepackage{pgfplots}
\usetikzlibrary{calc}
\usepgfplotslibrary{fillbetween}
\usetikzlibrary{patterns, positioning, arrows}
\usetikzlibrary{arrows.meta}
\usepackage{schemabloc}
\usetikzlibrary{decorations.markings}

\usepackage{algorithm}
\usepackage[noend]{algpseudocode}  
\usepackage[colorlinks=true,citecolor=blue,pdfpagemode=UseNone,pdfstartview=FitH]{hyperref}

\newif\ifRuodu
\ifRuodu
  
\fi

\renewcommand{\d}{\mathrm{d}}

\newcommand{\p}{\mathbb{P}}
\newcommand{\var}{\mathrm{var}}

\newcommand{\E}{\mathbb{E}}    
\newcommand{\R}{\mathbb{R}}    

\newcommand{\ES}{\mathrm{ES}}

  \newcommand{\id}{\mathds{1}}

\usepackage{setspace}

\theoremstyle{plain}
\newtheorem{theorem}{Theorem} 
\newtheorem{corollary}[theorem]{Corollary}
\newtheorem{lemma}[theorem]{Lemma}
\newtheorem{proposition}[theorem]{Proposition}

\theoremstyle{definition}

\newtheorem{example}{Example}

\theoremstyle{remark}

\def\laweq{\buildrel \mathrm{d} \over =}

\title{A new characterization of second-order stochastic dominance} 
\author{Yuanying Guan\thanks{Department of Mathematical Sciences and Department of Finance \& Real Estate, DePaul University, USA.  E-mail: \href{mailto:yguan8@depaul.edu}{yguan8@depaul.edu}.}  \and Muqiao Huang\thanks%
  {Department of Statistics and Actuarial Science,  University of Waterloo, Canada.
  E-mail: \href{mailto:m5huang@uwaterloo.ca}{m5huang@uwaterloo.ca}.} \and Ruodu Wang\thanks%
  {Department of Statistics and Actuarial Science,  University of Waterloo, Canada.
  E-mail: \href{mailto:wang@uwaterloo.ca}{wang@uwaterloo.ca}.}}

\pgfplotsset{compat=1.18}

\begin{document}
\maketitle
\begin{abstract}
We provide a new characterization of second-order stochastic dominance,
 also known as increasing concave order.
The result has an intuitive interpretation that adding  
a risk with negative expected value in adverse scenarios 
makes the resulting position generally  less desirable for risk-averse agents. 
A similar characterization is also found for convex order and increasing convex order.
The proof  techniques for the main result are based on properties of Expected Shortfall, a family of risk measures that is popular in banking and insurance regulation. 
Applications in risk management and insurance are discussed.

~ 

\end{abstract}

\section{Introduction}
 
 Second-order stochastic dominance (SSD), also known as increasing concave order, 
 is one of the most fundamental tools in stochastic comparison and decision making under risk (e.g., \cite{HR69} and \cite{RS70}). 
 For general treatments, we refer to the monographs \cite{MS02} and \cite{SS07}. 
 
This short paper provides a new characterization of SSD.  
 Let $L^1$ be the set of integrable random variables in an atomless probability space $(\Omega,\mathcal F,\p)$, which we fix throughout. 
 We first give the standard definitions for some stochastic orders. For $X,Y\in L^1$, we say that $X$ dominates    $Y$ 
\begin{enumerate}[(a)]
\item  in   {SSD}, denoted by 
  $X \ge_{\rm ssd} Y$, if $\E[u(X)]\ge \E[u(Y)]$ for all increasing concave functions $u$;
\item  in   {increasing convex order}, denoted by 
  $X \ge_{\rm icx} Y$, if $\E[u(X)]\ge \E[u(Y)]$ for all increasing convex functions $u$;
  \item in   {convex order}, denoted by 
  $X \ge_{\rm cx} Y$, if $\E[u(X)]\ge \E[u(Y)]$ for all   convex functions $u$.
\end{enumerate} 
Throughout the paper, ``increasing" is in the non-strict sense.

 In decision theory, $X$ and $Y$ in   comparison are usually interpreted as random payoffs or wealths. 
Instead, by interpreting  $X$ and $Y$ as losses,  SSD can be converted into increasing convex order, since $X\le_{\rm ssd} Y$ is equivalent to $-X\ge_{\rm icx} -Y$. 
  Increasing convex order is also known as stop-loss order in actuarial science; see e.g., \cite{DDGKV02}.
  We write $X\laweq Y$ if $X$ and $Y$ are identically distributed, which is precisely the symmetric part of each relation above.

Classic results on the representation of SSD are obtained by the celebrated work of \cite{S65} and \cite{RS70}.
  This representation result can be summarized as follows: For $X,Y\in L^1$, 
$
X\ge_{\rm ssd} Y
$
holds if and only if $
Y\laweq W + Z
$
for some $W,Z\in L^1$ such that $W\laweq X$ and 
\begin{align}\label{eq:classic}
\mbox{$\E[Z|W]\le0$.}
\end{align}
See Theorem 4.A.5 of \cite{SS07} for this result. 
 Condition \eqref{eq:classic} means that $(W,W+Z)$ forms a supermartingale.

The main result of this   paper is to provide a different representation of SSD,
where the corresponding condition for the additive payoff is weaker than \eqref{eq:classic}. 
\begin{theorem}\label{th:main}
 For any $X,Y\in L^1$,
 $
X\ge_{\rm ssd} Y
$
holds if and only if $
Y\laweq W + Z
$
for some $W,Z\in L^1$ such that $W\laweq X$ and 
\begin{align}\label{eq:new}
\mbox{$\E[Z|W\le x]\le 0$ for all relevant values of $x$.}\end{align}
\end{theorem}
In \eqref{eq:new}, relevant values of $x$ are those that satisfy $\p(W\le x)>0$. 

Note that  \eqref{eq:new} implies $\E[Z]\le 0$ by taking $x\to\infty$.
Condition \eqref{eq:new}
is clearly weaker than
\eqref{eq:classic}, because the latter
can be equivalently written as  
$\E[Z|W=x]\le 0$ for all almost every $x $  in the range of $W$ (here, the conditional expectations are chosen as a regular version). 
 If $Z$ is a function of $W$, then \eqref{eq:classic} is very restrictive, as 
it means   $Z\le 0$, whereas \eqref{eq:new} can hold for a wide range of models that does not require $Z\le 0$. 
Moreover, \eqref{eq:new} is much easier to check in practice, 
since the event $\{W\le x\}$ has a positive probability for  every relevant $x$,
whereas the event $\{W=x\}$   has zero probability for all $x$ when $W$ is continuously distributed. Two examples comparing \eqref{eq:classic} and \eqref{eq:new} are presented in Section \ref{sec:ex}. 
In Section \ref{sec:application}, we discuss applications of the new condition to risk management and insurance, including   stochastic improvers, marketable insurance contracts, and stop-loss premium calculation.

The main interpretation of Theorem \ref{th:main}
is that, for a risk-averse decision maker with random wealth $W$, adding a risk $Z$ with negative expectation in adverse scenarios, that is, when $W$ is small,
makes the resulting position $W+Z$ generally  less desirable than $W$; see Section \ref{sec:4} for more discussions. 

To prove Theorem \ref{th:main}, the main step is to justify 
$W+Z\le_{\rm ssd} W$, which we summarize in the following proposition.

\begin{proposition}\label{pr:main}
 For any $W,Z\in L^1$ satisfying  \eqref{eq:new}, 
 $W+Z \le_{\rm ssd} W$ holds. 
\end{proposition}

Up to the best of our knowledge, both Theorem \ref{th:main} and Proposition \ref{pr:main} are new to the literature.

Proposition \ref{pr:main} 
is closely related to the main result of \cite{B17},
where the author showed that for  $W$ taking values in $[0,1]$ and $Z$ taking values in $[-1,1]$,
if $\E[Z]=0$ and 
\begin{align}\label{eq:1}  \mbox{$\E[Z|W\ge x]\ge 0$ for all relevant values of $x$,}
 \end{align} 
then $W+Z \ge_{\rm cx} X$.
This result also follows from Corollary  3.3 of
  \cite{LLW16}.
Despite the close connection, 
there are several additional merits of our results and the proof approach.
First, our result works for both SSD and convex order,
whereas the condition \eqref{eq:1} of \cite{B17} does not generalize to SSD.
Indeed, \citet[Corollary 1]{B17} claimed that \eqref{eq:1} together with $\E[Z]\le 0$ yields $X+Z\le_{\rm ssd} X$, but $\E[Z]<0$ is not possible if  \eqref{eq:1} holds; thus SSD is not covered except for the case of convex order.
Similarly, results in \cite{LLW16} rely on the notion of expectation dependence (\cite{W87}; see Section \ref{sec:4}), but our condition \eqref{eq:new} is different from expectation dependence unless $\E[Z]=0$.
Second, our proof techniques are completely different from those of \cite{B17}. Our proof is much shorter, and it is based on risk measures, in particular, the Expected Shortfall (ES, also known as CVaR or TVaR), one of the most important risk measures in finance and insurance (\cite{MFE15}).
Thus, the proof is more accessible to  scholars in risk management. 
Third, our result is formulated on $L^1$ without any restriction on the range of the random variables,
and the proof argument is unified for all random variables without using discrete approximation or taking limits.

 \section{Proof of the main result}
 
Let us first define the risk measure ES used in our proof. ES at level $p\in [0,1)$ is defined by $$\ES_p(X) = \frac{1}{1-p}\int_p^1Q_X(t)\d t,~~X\in L^1,$$
where  $$Q_X(t)=\inf\{x\in \R: \p(X\le x)>t\}$$ is the right $t$-quantile of $X$ at $t\in (0,1)$. 
For $X\in L^1$,   denote by $\phi_X$ the function on $[0,1]$ given by $\phi_X(p)= (1-p) \ES_p(X)$ on $[0,1)$ and $\phi_X(1)=0$. 
We first state a few simple facts on  ES and the quantile function in the following lemma.

\begin{lemma}\label{lem:1}
For $X,Y\in L^1$, the following statements hold:
\begin{enumerate}[(i)]
\item $X\ge_{\rm icx} Y$ if and only if  $\ES_p(X)\ge \ES_p(Y)$ for all $p\in (0,1)$;
\item for $p\in (0,1)$,  $\ES_p(X) \ge \E[X|B]$ for any $B\in \mathcal F$ with $\p(B)=1-p$;

 \item for $p\in (0,1)$,    if $\p(X<Q_X(p))=p$, then 
 $\ES_p(X)=\E[X|X\ge Q_X(p)]$;

\item for $p\in (0,1)$, if $\p(X<Q_X(q))<\p(X<Q_X(p))$ for all $q\in (0,p)$, then $\p(X<Q_X(p))=p$;

\item the function $\phi_X$   is continuous and concave on $[0,1]$, and   its derivative  is $-Q_X(p)$ at almost every $p\in (0,1)$.
\end{enumerate}
\end{lemma}
\begin{proof}[Proof of Lemma \ref{lem:1}]
These properties are all well-known or straightforward to check. For (i), see  Theorem 4.A.3 of \cite{SS07}. 
 For (ii), see  Lemma 3.1 of \cite{EW15}.  
 For (iii), see   Lemma  A.7 of \cite{WZ21}. 
 Statement (iv) follows from Lemma 1 of \cite{GJW24}.
Statement (v) follows directly from the definition of  ES.
\end{proof}
 
 
%

\begin{proof}[Proof of Proposition \ref{pr:main}]
Note that  $W+Z \le_{\rm ssd} W$ is equivalent to $-W-Z\ge_{\rm icx} -W$, which we show below.
Write $X=-W$.  
By  (i) of Lemma \ref{lem:1}, 
it suffices to show   $\ES_p(X-Z)\ge \ES_p(X)$ for all $p\in (0,1)$. 
Denote by $P_X =\{p\in (0,1): \p(X<Q_p(X))=p\}$. 
For $p\in P_X $,  we have 
\begin{align*}
\ES_p(X-Z)& \ge \E[X-Z|X\ge  Q_X(p)]  \tag*{\footnotesize [by (ii)]}
\\&= \ES_p(X) - \E[Z|X\ge  Q_X(p)] \tag*{\footnotesize [by (iii)]}
\\&\ge \ES_p(X).\tag*{\footnotesize [by \eqref{eq:new}]}
\end{align*} 
 The argument is complete here if $X$ is continuously distributed, as in that case $P_X=(0,1)$. We continue with the case that the distribution of $X$ may have atoms. 
 
 Suppose that an interval $(a,b)\subseteq (0,1)$ does not intersect   $P_X$ (such intervals may not exist).
For any $p\in (a,b)$, 
let $p^*= \inf\{q\in (0,1): Q_X(q)=Q_X(p)\}$.
If $p^*=0$, then $q\mapsto Q_X(q)$ is  constant on $(0,p)$,
and since $p\in(a,b)$ is arbitrary, we have that $q\mapsto Q_X(q)$ is  constant on $(0,b)$. 
Next suppose $p^*>0$.
Since $q\mapsto Q_X(q)$ is right-continuous,
we have $Q_X(p^*) = Q_X(p)$.  
The definition of $p^*$ implies $Q_X(q)<Q_X(p^*)$ for any $q\in (0,p^*)$.
By (iv),   $\p(X< Q_X(p^*))= p^*$ and hence $p^*\in P_X$. This yields $p^*\le a$.
Therefore, $q\mapsto Q_X(q)$ is  constant on the interval $(a,b)$ in both cases.

By (v),  $\phi_X $ is linear on any interval outside $P_X$
and $\phi_{X-Z}$ is concave, and both are continuous.
We claim that 
these properties and $\phi_{X-Z} \ge \phi_X $ on $P_X$ imply that $\phi_{X-Z} \ge \phi_X $ holds   on $(0,1)$; 
see Figure \ref{FIG:pf} for an illustration. 
This would be sufficient for $X-Z\ge_{\rm icx}X$, and 
below we show this claim.  
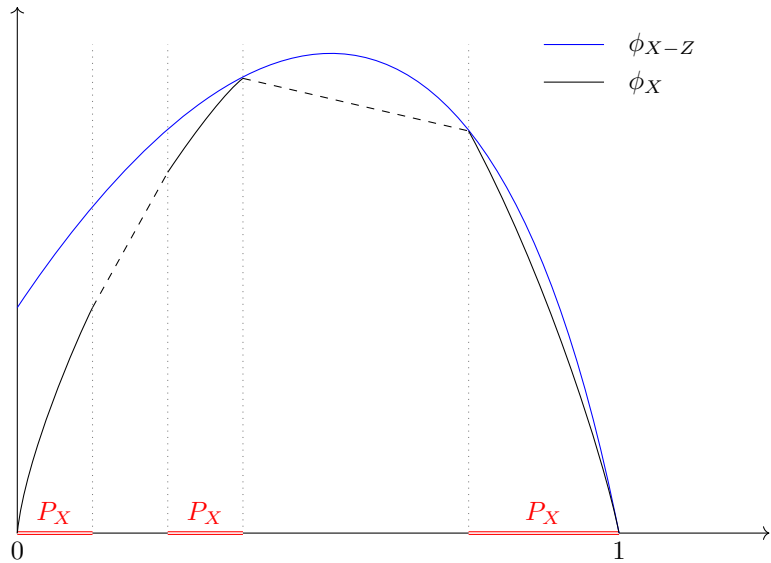
\begin{figure}[t]
\begin{center}
\begin{tikzpicture}
\draw[<->] (0,7) -- (0,0) -- (10,0);

\draw[gray,dotted] (1,0) -- (1,6.5);
\draw[gray,dotted] (2,0) -- (2,6.5);
\draw[gray,dotted] (3,0) -- (3,6.5);
\draw[gray,dotted] (6,0) -- (6,6.5);

\node[below] at (0,0) {0};
\node[below] at (8,0) {1};

\draw[red, double] (0,0) -- (1,0);
\draw[red, double] (2, 0) -- (3, 0);
\draw[red, double] (6,0) -- (8,0);
\node[red, above] at (0.5, 0)  {$P_X$};
\node[red, above] at (2.5, 0)  {$P_X$};
\node[red, above] at (7, 0)  {$P_X$};

\draw[blue] (7, 6.5) -- (7.8, 6.5);
\node[right] at (8,6.5) {$\phi_{X-Z}$};
\draw (7, 6) -- (7.8, 6);
\node[right] at (8,6) {$\phi_{X}$};

\draw  [blue] (0,3) .. controls (2, 6) and (6, 10) .. (8,0);

\draw   (0,0) .. controls (0.1, 1) and (0.8, 2.6) .. (1,3);
\draw   [dashed] (1,3) -- (2,4.8);
\draw   (2,4.8) .. controls (2.4, 5.4) and (2.8, 5.9) .. (3,6.05);
\draw  [dashed] (3,6.05) -- (6,5.35);
\draw   (6,5.35) .. controls (7, 3.5) and (7.8, 1).. (8,0);
\end{tikzpicture} 
\caption{An illustration of $\phi_{X-Z} \geq \phi_X$: after the inequality is shown to hold on $P_X$,  it also holds outside $P_X$  due to concavity of $\phi_{X-Z} $ and piece-wise linearity of $\phi_X$ (dashed lines).}
    \label{FIG:pf}
    
\end{center} 
\end{figure}

Suppose   that there exists $p\in (0,1)$ such that $\phi_{X-Z}(p)<\phi_X(p)$. 
Since both $\phi_{X-Z}$ and $\phi_X$ are continuous, there exists a neighbourhoold of $p$ on which   $\phi_{X-Z}<\phi_X$.
Let $a=\inf\{q\in (0,1): \phi_{X-Z}<\phi_X \mbox{~on~$(q,p]$} \}$  
and 
$b=\sup\{q\in (0,1): \phi_{X-Z}<\phi_X \mbox{~on~$[p,q)$} \}$.
If $a>0$, then 
by continuity we have $\phi_{X-Z}(a)=\phi_{X}(a)$
and $\phi_{X-Z}(b)=\phi_{X}(b)$, noting that $\phi_{X-Z}(1)=\phi_{X}(1)=0$.
If $a=0$, then 
$\phi_{X-Z}(a)\ge \phi_{X}(a)$ because  
$$\phi_{X-Z}(0)=\E[X]-\E[Z]\ge\E[X]=\phi_X(0),$$
where $\E[Z]\le 0$ is guaranteed by \eqref{eq:new}.
Therefore, in either case, \begin{equation}
    \phi_{X-Z}(a)\ge \phi_X(a)
\mbox{~~and~~} 
\phi_{X-Z}(b)\ge  \phi_X(b).\label{eq:details}
\end{equation}
Note that $(a,b)$ does not intersect $P_X$ since $\phi_{X-Z}\ge \phi_X$ on $P_X$.
Since $\phi_{X-Z}$ is concave and $\phi_X$ is linear on $(a,b)$, \eqref{eq:details} implies $\phi_{X-Z}\ge \phi_X$ on $(a,b)$; see the areas with dashed lines in Figure \ref{FIG:pf}. This yields a contradiction. 
\end{proof} 

\begin{proof}[Proof of Theorem \ref{th:main}]
The ``if" statement follows from Proposition \ref{pr:main} via
$X\laweq W \ge_{\rm ssd} W+Z \laweq Y$.
The ``only if" statement follows from the fact that \eqref{eq:classic} is stronger than \eqref{eq:new}, and thus via the representation mentioned in the introduction (Theorem 4.A.5 of \cite{SS07}).
\end{proof}

\section{Convex order and increasing convex order}
\label{sec:3}
We present a few immediate corollaries of Theorem \ref{th:main}  on increasing convex order and  convex order, commonly used in risk management and actuarial science.

\begin{corollary}\label{co:1}
For any $X,Y\in L^1$,
 $
X\le_{\rm icx} Y
$
holds if and only if $
Y\laweq W + Z
$
for some $W,Z\in L^1$ such that $W\laweq X$ and 
$$
\mbox{$\E[Z|W\ge x]\ge0$ for all relevant values of $x$.}$$
 
\end{corollary}
\begin{proof}
This corollary follows by applying Theorem \ref{th:main} to
the  relation $-X\ge_{\rm ssd} -Y$, and let $-Y \laweq -W-Z$ with $\E[-Z|-W\le -x]\le 0$ for all relevant values of $x$.
\end{proof}

 \begin{corollary}\label{co:2}
For any $X,Y\in L^1$,
the following are equivalent.
\begin{enumerate}[(i)]
\item 
 $
X\le_{\rm cx} Y;
$
\item $X\ge _{\rm ssd} Y$ and $\E[X]=\E[Y]$;
\item $
Y\laweq W + Z
$
for some $W,Z\in L^1$ such that $W\laweq X$ and 
 $\E[Z|W]=0 $. 
\item $
Y\laweq W + Z
$
for some $W,Z\in L^1$ such that $W\laweq X$ and 
 $$
\mbox{$\E[Z]=0$ and $\E[Z|W\le x]\le0$ for all relevant values of $x$;} $$

\item $
Y\laweq W + Z
$
for some $W,Z\in L^1$ such that $W\laweq X$ and 
 $$
\mbox{$\E[Z]=0$ and $\E[Z|W\ge x]\ge0$ for all relevant values of $x$.}$$
 \end{enumerate}
\end{corollary}
 \begin{proof}
The equivalence between (i), (ii) and (iii) is well known; see e.g., \citet[Theorems 3.A.4 and 4.A.35]{SS07}.
The equivalence between (ii) and (iv) follows from Theorem \ref{th:main}.
The equivalence between (iv) and (v) follows by noting that  
for $Z$ with $\E[Z]=0$, 
$$\E[Z|W\le x]\le 0 ~\forall x \iff \E[Z|W> x]\ge  0 ~\forall x \iff \E[Z|W\ge x]\ge  0 ~\forall x,$$
where the last equivalence is argued by  $\lim_{x\uparrow y}\{W> x\}=\{W\ge y\}$ and $\lim_{x\downarrow y}\{W\ge x\}=\{W>y\}$. 
 \end{proof}
 The implication (v) $\Rightarrow$ (i)  in Corollary \ref{co:2} is also obtained by \cite{LLW16} and  \cite{B17}.
To compare (iv) and (v) with the classic characterization (iii),   if $Z$ is a function of $W$, $\E[Z|W]=0$ does not hold except for the trivial case $Z=0$, whereas $\E[Z|W\le x]\le 0$ and $\E[Z|W\ge x]\ge 0$ in (iv) and (v) can hold for many models of $Z$, thus allowing for   flexibility in applications.

\section{Two illustrative examples}
\label{sec:ex}

We now illustrate our results with two simple examples,
one with  Gaussian distributions
and one with Bernoulli distributions.
The purpose here is to compare the classic condition \eqref{eq:classic} with our condition \eqref{eq:new}, and to see how much more flexibility \eqref{eq:new} offers. 
Note that both \eqref{eq:classic} and \eqref{eq:new}
are sufficient conditions for $W+Z\le_{\rm ssd} W$, but they may not be necessary.
In both examples below,  \eqref{eq:new} allows a larger range of parameters than \eqref{eq:classic}. 

\begin{example}[Gaussian distributions]
\label{ex:1}
It is well known that for two normal random variables $X\sim \mathrm{N}(\mu_X,\sigma_X^2)$ and $Y\sim \mathrm{N}(\mu_Y,\sigma_Y^2)$, $X\ge _{\rm ssd} Y$ holds if and only if 
\begin{align}\label{eq:normal}
\mbox{$\mu_X\ge \mu_Y \mathrm{~~and~~} \sigma^2_X\le \sigma^2_Y$}
\end{align}
see e.g., Example 4.A.46 of \cite{SS07}. 
Suppose that $(W,Z)$ is jointly Gaussian with mean vector $(\mu_W,\mu_Z)$ and covariance matrix 
$$
\begin{pmatrix}
\sigma_W^2 & \rho \sigma_W \sigma_Z\\
\rho \sigma_W \sigma_Z&  \sigma_Z^2  
\end{pmatrix},
$$
where $\rho$ is the correlation coefficient. 
Since SSD is invariant up to location-scale transforms, 
 we assume $\mu_W=0$ and $\sigma_W=1$ without loss of generality.
We analyze values of  the parameters $\mu_Z\in \R$, $\sigma_Z>0$ and $\rho\in[-1,1]$ obtained from $W+Z\le_{\rm ssd} W$, \eqref{eq:classic}, and \eqref{eq:new}, respectively. 
\begin{enumerate}[(a)]
\item Since $\E[W+Z] = \mu_Z$ and  $\var(W+Z) = 1+\sigma_Z^2 + 2 \rho  \sigma_Z$, we have from \eqref{eq:normal} that 
$W+Z\le_{\rm ssd} W$ if and only if  
$\mu_Z\le 0$ and $\rho\ge -\sigma_Z/2$.
\item  
Using the conditional distribution of the bivariate Gaussian distribution, we have  $\E[Z|W]=\mu_Z+ \rho\sigma_Z W$. Condition \eqref{eq:classic} is $\mu_Z+ \rho\sigma_Z W\le 0$.
Since the support of $W$ is the real line,
this means $\mu_Z\le 0$ and $\rho=0$.
\item  Condition \eqref{eq:new}
is $\mu_Z+ \rho\sigma_Z \E[W|W\le x]\le 0$ for all $x\in \R$. If $\rho\ge 0$, this holds true if and only if $\mu_Z\le 0$,
because  $\sup_{x\in \R} \E[W|W\le x]=\E[W]= 0$.
If $\rho<0$, this does not hold for any $\mu_Z$ and $\sigma_Z$, since $\E[W|W\le x]$ is unbounded from below. 
\end{enumerate}
In this example, condition \eqref{eq:new} is more flexible than \eqref{eq:classic}, although it is stronger than the equivalent condition for SSD.  
  Table \ref{tab:1} summarizes these observations.
  
\begin{table}
\renewcommand{\arraystretch}{1.5}
\begin{center}
\begin{tabular}{ l|l  }
sufficient condition \eqref{eq:classic} &  $\mu_Z\le 0$ and $\rho=0$   \\\hline 
our sufficient condition \eqref{eq:new} & $\mu_Z\le 0$ and $\rho\ge 0$   \\  \hline 
equivalent condition for SSD & $\mu_Z\le 0$ and $\rho\ge -\sigma_Z/2$   \\  \hline 
\end{tabular}
\end{center}
\caption{Conditions for $W+Z\le_{\rm ssd} W$ in Example \ref{ex:1} (Gaussian)}
\label{tab:1}
\end{table}
\end{example}

\begin{example}[Bernoulli distributions]
\label{ex:2}
Consider a Bernoulli random variable $W$ with parameter $1/2$ and a random variable $Z$ distributed as $ W-c$ for a constant $c\in \R$. 
Let $\rho$ be the correlation coefficient of $(W,Z)$.
Note that $\rho$ fully determines the joint distribution of $(W,Z)$, where the only degree of freedom is $\p(W=Z+c=1)=(1+\rho)/4$.
We analyze values of  the parameters $c \in \R$  and $\rho\in[-1,1]$ obtained from $W+Z\le_{\rm ssd} W$, \eqref{eq:classic}, and \eqref{eq:new}, respectively. 
\begin{enumerate}[(a)]
\item 
First, $\p(W+Z=2-c)=\p(W+Z= -c)=(1+\rho)/4$
and 
$\p(W+Z=1-c)=(1-\rho)/2$. 
For  $W+Z\le_{\rm ssd} W$ to hold, it is necessary and sufficient that $\E[(W+Z)\wedge t]\le \E[W\wedge t]$ for all $t\in \R$, where $a\wedge b=\min\{a,b\}$.
Since $W+Z$ takes values only at three points $-c,1-c,2-c$, it suffices to check these three points.
Checking the point  $t=2-c$ yields $c\ge 1/2$.
Checking the point  $t=1-c$ yields $\rho\ge 1-2c$.
Checking the point  $t= -c$ yields $c\ge 0$.
Therefore, the equivalent condition is 
$c\ge 1/2$ and $\rho\ge 1-2c$. 
\item  We can directly compute $\E[Z|W]= (1-\rho)/2-c+ \rho W$. Therefore, condition \eqref{eq:classic} means 
$(1-\rho)/2-c  \le 0$ if $\rho\ge 0$
and
$(1-\rho)/2-c + \rho \le 0$ if $\rho<0$.
Putting the two cases together, it is  $1-2c \leq \rho \leq 2c-1$ (which implies $c\ge 1/2$).
\item Note that $\E[W|W\le x]$ for relevant $x$
 takes value $0$ or $1/2$. Hence, condition \eqref{eq:new}
 is $ (1-\rho)/2-c\le 0$
 and $ (1-\rho)/2-c+ \rho/2\le 0$,
 and thus $c\ge 1/2$ and $\rho\ge 1-2c$.
\end{enumerate} 
In this example, condition \eqref{eq:new} is more flexible than \eqref{eq:classic}, and it turns out to be necessary and sufficient.  
  Table \ref{tab:2} summarizes these observations.
\begin{table}
\renewcommand{\arraystretch}{1.5}
\begin{center}
\begin{tabular}{ l|l  }
sufficient condition \eqref{eq:classic} &  $c\ge 1/2$ and $1-2c \leq \rho \leq 2c-1$  \\\hline 
our sufficient condition \eqref{eq:new} & $c\ge 1/2$ and $  1-2c \le \rho$   \\  \hline 
equivalent condition for SSD & $c\ge 1/2$ and $  1-2c \le \rho$  \\  \hline 
\end{tabular}
\end{center}
\caption{Conditions for $W+Z\le_{\rm ssd} W$ in Example \ref{ex:2} (Bernoulli)}
\label{tab:2}
\end{table}
\end{example}

In both examples, the condition via \eqref{eq:new} and the equivalent SSD condition  are both quite intuitive: the mean of $Z$ is less than $0$ (this is necessary for the SSD relation), and the correlation of $(W,Z)$ cannot be too small, as negative correlation reduces the aggregate risk. 
The condition via \eqref{eq:classic} does not have this interpretation. 

For a bi-atomic distribution of $Z$ on two arbitrary points (more general than $Z$ in Example \ref{ex:2}), we can check that condition $\eqref{eq:new}$ is equivalent to $Z+W \leq_\mathrm{ssd} W$ when the spread of $Z$ is less than or equal to 1, and otherwise it is stronger than necessary. We omit these calculations.

\section{ Risk management and insurance applications}
\label{sec:application}

In this section, we provide three applications of our main results: stochastic improvers, insurance marketability, and stop-loss premium calculation. 

\subsection{Risk reducers and stochastic improvers}

As a risk management tool, 
the concept of
 risk reducers is introduced  by \cite{CDLT14} and further studied by \cite{HTZ16}. 
A \emph{risk reducer} for $X\in L^1$ (\cite{CDLT14}) is an additive payoff $Z\in L^1$ such that 
$X+Z\le_{\rm cx} X+\E[Z]$. 
Intuitively, risk reducers are random payoffs $Z$ that make the combined payoff $X+Z$ less risky than the original wealth $X$ adjusted by the mean of $Z$. 

Inspired by this, we define a \emph{stochastic improver} for $X\in L^1$, that is, an additive payoff $Z\in L^1$ such that $X+Z\ge_{\rm ssd} X$.
Recall that risk-averse expected utility agents are modelled by increasing concave utility functions.
The intuition of a stochastic improver is that every risk-averse expected utility agent would prefer $X+Z$ over $X$, and thus the additive payoff $Z$ improves the utility of the agent with random wealth $X$.

Note that $X+Z\ge_{\rm ssd} X+\E[Z]$ is equivalent to $X+Z\le_{\rm cx} X +\E[Z]$.
Hence, for $Z$ with $\E[Z]=0$, it is a risk reducer if and only if it is a stochastic improver.
However, for random variables with non-zero mean, these two concepts are generally incompatible. 
If $\E[Z]\ge 0$, then a risk reducer is necessarily a stochastic improver,
because 
$$X+Z\ge_{\rm ssd} X + \E[Z] \ge_{\rm ssd} X.$$
However, the converse is not true; a stochastic improver  need not be a risk reducer. For instance, for any nonnegative  $X\in L^2$ with positive variance,
we have $X+X\ge_{\rm ssd} X$ but $\var(X+X)>\var(X+\E[X])$;
therefore $X$ is a stochastic improver for itself but not a risk reducer; in fact, $X+X\ge_{\rm cx} X+\E[X]$ holds (Theorem 3.A.17 of \cite{SS07}), the opposite of being a risk reducer.  
Moreover, a risk reducer $Z$ is not a stochastic improver if $\E[Z]<0$.

Denote by $\mathcal S_X$ the set of all stochastic improvers for $X\in L^1$, and let
$$
\mathcal N_X=\{Z\in L^1:  \mbox{$\E[Z|X+Z\le x]\ge 0$ for all relevant values of $x$}\}.
$$ 
The definition of SSD implies that the set $\mathcal S_X$ is convex. 
The following result connects the above two sets by using our main result. 
\begin{proposition}
\label{prop:improve}
For $X\in L^1$, $\mathcal N_X \subseteq \mathcal S_X.$
\end{proposition}
\begin{proof}
It suffices to verify that for $Z\in \mathcal N_X$, $X+Z\ge_{\rm ssd} X$. 
Let $W=X+Z$.
Using Proposition \ref{pr:main}, the condition $\E[-Z|W\le x]\le 0$ for all relevant $x$ is sufficient for 
$W\ge_{\rm ssd} W-Z$, which is $X+Z\ge_{\rm ssd} X$.
\end{proof}

One may wonder whether the converse statement to Proposition \ref{prop:improve} also holds, that is, $\mathcal S_X= \mathcal N_X$.
The quick answer is negative.
As we see in Table \ref{tab:1}, for two standard Gaussian random variables $W$ and $Z$,
$W-Z\le_{\rm ssd} W$ if and only if the correlation coefficient between $W$ and $Z$ is smaller than or equal to $ 1/2$.
By writing $X=W-Z$, 
the above condition is equivalent to $Z\in \mathcal S_X$.
On the other hand, 
$Z\in \mathcal N_X$ if and only if the correlation coefficient between $W$ and $Z$ is nonpositive. 
Therefore,  $Z\in \mathcal S_X$ but $Z\not \in \mathcal N_X$.

There is a special setting in which $\mathcal S_X$ and $\mathcal N_X$ coincide. 
A random vector $(X,Y)$ is \emph{comonotonic} if there exist increasing functions $f$ and $g$ such that $X=f(X+Y)$ and $Y=g(X+Y)$ almost surely. 
\cite{HTZ16} studied risk reducers $Z$ when $(X,X+Z)$ is comonotonic. Below we obtain a characterization of stochastic improvers under the same assumption of comonotonicity.  

\begin{proposition}
\label{prop:improve2}
Let $X,Z\in L^1$ be such that $(X,X+Z)$ is comonotonic.
Then  $Z\in \mathcal N_X $ if and only if $Z\in \mathcal S_X.$
\end{proposition}
\begin{proof}
The ``only-if" statement follows from Proposition \ref{prop:improve}. We show the ``if" statement below. Suppose   $Z\not \in \mathcal N_X$. 
By definition of $\mathcal N_X$, there exists $x\in \R$ such that 
$\E[Z|X+Z\le x]<0$ and $\p(X+Z\le x)>0$.
Let $A=\{X+Z\le x\}$ and $p=1-\p(A)$.
Since $\E[Z]\ge 0$ as required by $Z\in \mathcal S_X$, we have $\p(A)\in (0,1)$. 
It follows that
\begin{equation}\label{eq:improve}
\E[X+Z|A]=\E[X|A]+\E[Z|A]<\E[X|A].
\end{equation}
Let $W=-X$.
Since $(X,X+Z)$ is comonotonic, 
$(W,\id_{A })$ is also comonotonic. 
Therefore, for almost every $\omega\in A $ and $\omega'\in A^c $, 
we have $W(\omega)\ge W(\omega')$.
Such 
$A$ is called a  $p$-tail event of $W$ by \cite{WZ21}; intuitively, it is a set on which $W$ takes larger values than on its complement. 
By definition, $A$ is also a p-tail event of $-X-Z=W-Z$.
 Lemma A.7 of \cite{WZ21} gives
$\E[W|A]= \ES_p(W)$ and $\E[W-Z|A]= \ES_p(W-Z)$.
Putting the above observations with \eqref{eq:improve}, we have 
$\ES_p(W-Z)> \ES_p(W)$, and hence by Lemma \ref{lem:1} part (i),
$W-Z\le_{\rm icx} W$ cannot hold.
This means that
$X+Z\ge_{\rm ssd} X$ does not hold, and $Z\not \in \mathcal S_X$.
Thus, $Z\in \mathcal S_X$ implies $Z\in \mathcal N_X$.
\end{proof}

 Applying Proposition \ref{prop:improve2} to $Z-\E[Z]$ with  $(X,X+Z)$  comonotonic, we obtain 
that the following conditions are equivalent:
\begin{enumerate}[(a)]
\item 
$Z-\E[Z]$ is a stochastic improver for $X$;
\item 
$Z$ is a risk reducer for $X$;
\item   $\E[Z|X+Z\le x]\ge \E[Z]$ for all relevant $x$.
\end{enumerate}
This result can be compared with \citet[Theorem 3.2]{HTZ16}, which states that for $Z$ that is $\sigma(X)$-measurable with $(X,X+Z)$ comonotonic, 
 (b) is equivalent to \begin{enumerate}[(a)]
\item[(d)]  $\E[Z|X \le x]\ge \E[Z]$ for all relevant $x$.
\end{enumerate} The   condition (d) is called negative expectation dependence of $Z$ on $X$ (\cite{W87}). 
Generally,
 the two conditions 
(c)  and (d)
 are not equivalent even if $(X,X+Z)$ is comonotonic; for instance, if $X$ is a constant, then (d) always holds true but (c) never holds true unless $Z$ is also a constant.
 Nevertheless, when $Z$ is $\sigma(X)$-measurable,
 (d) implies (c) because the set of events  $\{\{X\le x\}: x\in \R \}$ 
contains $\{\{X+Z\le y\}: y\in \R \}$ in this case.
Therefore, 
Proposition \ref{prop:improve2} implies Theorem 3.2 of \cite{HTZ16}
as a special case when $Z$ is $\sigma(X)$-measurable and condition (d) holds.

%
%

An example of a stochastic improver satisfying the conditions in Proposition \ref{prop:improve2} is the purchase of a protective put in a Black--Scholes financial market.  
Below, a random vector $(X,Y)$ is \emph{counter-monotonic} if $(X,-Y)$ is comonotonic.  
\begin{example}[Protective put]
Consider a continuous-time financial market model with $0$ interest rate on a time interval $[0,T]$.
For simplicity, we will assume a Black--Scholes market with a stock price process $(X_t)_{t\in [0,T]}$ that has a nonpositive return rate and constant volatility.
The assumption of nonpositive return rate is unusual, but it is needed for the analysis below. 
For details on the Black--Scholes market model used in this example, see \cite{S04}.
This market is complete, and thus any payoff can be priced with a risk-neutral probability measure $Q$. 
For $t\in [0,T]$, 
let $\mathcal F_t=\sigma(X_s: s\le t)$ and $P_t$ be the time-$t$ price of a put option that gives the payoff $(K-X_T)_+$ at maturity $T $, where $K>0$ represents the strike price. 
The strategy of holding the stock and purchasing the put option is called a protective put. 
The Black--Scholes formula gives that  $P_t$ is a decreasing function of $X_t$ for each $t\in [0,T]$. 
Let $Z_t=P_t-P_0$, that is, the time-$t$ value of purchasing the put option at time $0$.
 Girsanov's theorem gives the explicit formula of $Y_t:=\E[\d Q/\d \p|\mathcal F_t]$, 
 which guarantees that $Y_t$ is an increasing function of $X_t$ under the assumption of nonpositive return rate.  Hence, $$P_0=\E^Q[P_t] = \E\left[  \E\left[ \frac{\d Q}{\d \p}|\mathcal F_t\right]P_t \right]\le \E\left[ \frac{\d Q}{\d \p}\right]  \E[P_t]  = \E[P_t],$$
where the inequality is due to the counter-monotonicity of $(P_t,Y_t)$.
Moreover, $ X_t+Z_t $ is an increasing function of $X_t$ for $t\in [0,T]$, which can be seen from e.g., the put-call parity. Using counter-monotononicity of $(P_t,X_t+Z_t)$, which implies negative expectation dependence of $P_t$ on $X_t+Z_t$, we have, for any relevant $x\in \R$,  
$$
\E[Z_t |X_t+Z_t\le x] =\E[P_t - P_0 |X_t+Z_t\le x] \ge \E[P_t]-P_0 \ge 0.
$$
Therefore, the conditions in Proposition \ref{prop:improve2} hold for $(X_t,Z_t)$, and hence $Z_t$ is a stochastic improver for $X_t$. 
In conclusion, for any risk-averse expected utility agent, entering a protective put at time $0$ improves the expected utility of the future payoff at each time spot up to the maturity. Recall that this conclusion only holds if the asset has a nonpositive return rate, making it undesirable for most investors. The assumption of nonpositive return is replaced by that of a nonnegative return if the investor has a short position of the stock. In  that case, a call option is a stochastic improver, following the above argument. 
\end{example}

\subsection{Widely marketable insurance contracts}

Next, we consider an insurance market. 
Let $L^1_+$ be the set of nonnegative random variables in $L^1$, and elements in $L^1_+$ represent insurable losses in this section.
\cite{CDLT14} introduced the concept of 
universal markability.
An \emph{indemnity schedule} is a function $I:\R_+\to \R_+$ satisfying $0\le I(x)\le x$ for each $x\ge 0$.
An  indemnity schedule $I$ is 
\emph{universally marketable}  if for any   $X\in L^1_+$,   $w\in \R$,  and increasing concave utility function $u$, a solution $P^*$ to the equation 
\begin{equation}\label{eq:price}
\E[u(w-X+I(X)-P)] = \E[u(w-X)],~~~~P\in \R,
\end{equation}
satisfies $P^*\ge \E[I(X)]$. 
Intuitively, it means that every risk-averse expected utility agent with insurable loss $X$ would accept to purchase the insurance contract with payoff $I(X)$ at some price higher than or equal to the mean of $I(X)$; the insurance price being no less than the mean of the payoff is a natural requirement for the insurance provider to participate (see \cite{A63}); later we will discuss a few examples where this is violated. 
\citet[Theorem 3]{CDLT14} showed that an indemnity schedule  is universally marketable if and
only if it is $1$-Lipschitz.

Below, we offer a different angle. In an insurance market, the indemnity schedule $I$ and the loss $X$ are not separately considered; for instance, the indemnity schedule should be different for property insurance and for health insurance. 
Therefore, instead of looking for $I$ that is marketable for all $X$,
it is natural to look for $I$ that is marketable for a specific $X$. 
Moreover, the insurance company may be concerned about a different premium principle than the expected value (see e.g., \cite{D90} and \cite{WYP97}). We let $P_0 \ge 0$ represent the minimum acceptable price for the insurer 
for the indemnity $I$; in the previous setting it is $P_0=\E[I(X)]$.

To incorporate the above two features, we say that 
an  indemnity schedule $I$ is 
\emph{widely marketable for $(X,P_0)\in L^1_+\times \R_+$} if for any $w\in \R$  and increasing concave utility function $u$, a solution $P^*$ to \eqref{eq:price} satisfies $P^*\ge  P_0$. 
Here, we choose the word ``widely" to reflect that this requirement is less general than universal marketability (which holds for all $X$), but   it is still quite broad, as it applies to all risk-averse expected utility agents.
Our main results allow us to study this property for flexible choices of $P_0$, which is summarized in the next proposition. 
\begin{proposition}
\label{prop:marketable}
Let $I$ be an indemnity schedule, $X\in L^1_+$,  and $P_0\in \R_+$.
If $\E[I(X)|X-I(X)\ge x] \ge P_0$ for all relevant $x$, then $I$ 
is  {widely marketable} for $(X,P_0) $.
\end{proposition} 
\begin{proof}
Since the utility function $u$ is increasing and concave,   a solution $P^*\ge P_0$ to \eqref{eq:price} exists for every $u$ if
$
-X+I(X)-P_0\ge_{\rm ssd} -X. 
$
By Proposition \ref{pr:main} with $W=-X+I(X)-P_0$ and $Z=P_0-I(X)$, a sufficient condition for the desired SSD relation is 
$$\E[P_0-I(X)|-X+I(X)-P_0\le x]\le 0 \mbox{~ for all relevant values of $x$.}$$
By rearranging terms, the above condition is equivalent to $\E[I(X)|X-I(X)\ge x] \ge P_0$ for all relevant $x$.
\end{proof}

Due to the specification of $X$ and $P_0$, the indemnity $I$ in Proposition \ref{prop:marketable} need not be  continuous (see the example below), thus a wider class of indenmity schedules than those in Theorem 3 of \cite{CDLT14} can be included. 

\begin{example}[Fixed idenmity plan]
Let the indenmity schedule $I$ be given by $I(x)=\id_{\{x\ge 1\}}$. That is, an pre-determined payment of $1$ is paid if the loss  reaches or exceeds $1$.
This kind of contract is called a fixed indemnity plan in health insurance (e.g., a fixed amount is paid upon hospital admission). 
Clearly, $I$ is not continuous, but it satisfies all conditions to be an indemnity schedule. 
Suppose that $X$ is exponentially distributed with mean $1$.  
We can compute, for $x\in [0,1]$,
\begin{align*}
\E[I(X)|X-I(X) \ge x] &=\p(X\ge 1 |X-I(X) \ge x)
\\&= \frac{\p(X\ge 1\mbox{~and~} X-I(X)\ge x) }{\p(X-I(X)\ge x)}
\\&= \frac{\p(X\ge 1 + x)}{ \p(X\in [x,1)\cup [1+x,\infty))}
\\&= \frac{e^{-(1+x)}}{ e^{-(1+x)}  + e^{-x}-e^{-1}}
 = \frac{1}{ 1  + e -e^x}\ge e^{-1}. 
\end{align*}
Moreover, for $x>1$, $\E[I(X)|X-I(X) \ge x] =1$.
Therefore, with any $P_0\in[0, e^{-1}]$, $I$ is widely marketable for $(X,P_0).$ In particular, $I$ is widely marketable for $(X,\E[I(X)])$, by noting $\E[I(X)]=e^{-1}$. \end{example}

If $P_0>\E[I(X)]$, then the condition 
$\E[I(X)|X-I(X)\ge x] \ge P_0$ in Proposition \ref{prop:marketable}  cannot hold for all relevant $x$. Therefore, 
if the insurance company  charges more than the expected value of the insurance payment, 
then its contract cannot be attractive to all risk-averse expected utility agents (although it may still be attractive to a subset of such agents).
This is because risk-averse agents include risk-neutral ones, who do not want to pay anything more than the expected insurance payment and are not the typical insurance buyers. 
This may be seen as a limintation of the applicability of Proposition \ref{prop:marketable}. 
The same limintation applies to the formulation of \cite{CDLT14}, which relies on the stronger condition $P_0=\E[I(X)]$.
On the positive side, the additional flexibility of $P_0$ provided by Proposition \ref{prop:marketable}
allows an insurance company to quantitatively understand how to attract all risk-averse expected utility agents for a particular product, if they wish to, by lowering their premium to the level that the condition in Proposition \ref{prop:marketable} holds. 
This is relevant in the contexts  of many different insurance  products, commercial promotions, and 
 government subsidized insurance. In each of these contexts, the insurance company may have an incentive to let the premium go below the expected insurance payment.

\subsection{Stop-loss premium calculation}

The SSD relation is also closely related to the stop-loss premium in insurance.  
Let $X\in L^1_+$ be an insurable loss. 
For a deductible level $d\ge 0$, 
the insurance contract 
that pays $(X-d)_+$
 is called a stop-loss insurance contract, which is  a popular form of  insurance coverage; \cite{A63} showed that the stop-loss contract is the optimal form for a risk-averse insured and a risk-neutral insurer under general conditions. 
The stop-loss premium of $X$ with deductible $d$ is then defined as $\E[(X-d)_+]$, widely studied in actuarial science;  see e.g., \cite{DV98} and \cite{DDGKV02}. 
Our results imply a simple relation on the stop-loss premiums of two random losses. 
\begin{proposition}\label{prop:stl}
If $X\in L^1_+$ and $Z\in L^1$ satisfy $\E[Z|X\ge x]\ge0$ for all relevant values of $x$, 
then for any deductible $d\ge 0$, $X+Z$ has a larger stop-loss premium than $X$.
\end{proposition}
Proposition \ref{prop:stl} is a straightforward consequence of  Corollary \ref{co:1} and 
the well-known fact that the partial order over $L^1_+$ induced by stop-loss premiums at all deductible levels 
is equivalent to increasing convex order; see e.g., \cite{DDGKV02}.

\section{Further discussions}
\label{sec:4}

In this section we discuss some issues related to our results. 
We first present an extension of the representation in Theorem \ref{th:main}.

In the formulation of Theorem \ref{th:main}, 
in addition to the additive payoff $Z$, 
we relied on an extra random variable $W\laweq X$ satisfying $Y\laweq W+Z$ instead of directly using  $W=Y-Z$. 
This is needed in the classic representation  (called the Strassen theorem); see
 Theorem 4.A.5 of \cite{SS07}.
 The technical reason for such a construction is that the existence of $W$ satisfying certain distributional requirements depends on the choice of $Y$ and the underlying probability space.
 For instance,   an independent noise to $Y$ does not exist  if the $\sigma$-algebra of $Y$ is equal to $\mathcal F$. 
 
 In a recent paper, \cite{NWZ22}   established a refinement of  the Strassen theorem by allowing   $W=Y-Z$ in an arbitrary atomless probability space.  This  refinement, based on the theory of martingale optimal transport, is highly non-trivial.
  Using this result, we obtain a representation without involving an additional random variable $W$.  
Our results in the previous sections are presented in their current forms for consistency with the classic  theorem of Strassen in its most familiar form.
 In what follows, $X\ge_{\rm st} Y$ means $\p(X>x)\ge \p(Y>x)$ for all $x\in \R$. 

\begin{theorem}
\label{th:strong}
 For any $X,Y\in L^1$, 
 $
X\le_{\rm cx} Y
$ 
holds if and only if $
X\laweq Y-Z  
$
for some $ Z\in L^1$ such that $\E[Z]=0$ and
\begin{align}\label{eq:new5}
\mbox{$\E[Z|Y-Z\le x]\le 0$ for all relevant values of $x$.}\end{align}
Moreover, 
  $
X\ge_{\rm ssd} Y
$
holds if and only if $
X\ge_{\rm st} Y-Z  
$
for some $ Z\in L^1$ such that  \eqref{eq:new5} holds.
\end{theorem}
\begin{proof} The second equivalence is a direct consequence of the first one by decomposing $\ge_{\rm ssd}$ into $\ge_{\rm st}$ and $\le_{\rm cx}$; see Theorem 4.A.6 of \cite{SS07}.
Below we only show the first equivalence. 
The 	``if" statement follows from Corollary \ref{co:2} with $W=Y-Z$. For the ``only if" statement,  
by Theorem 3.1 of \cite{NWZ22}, 
there exists $W\in L^1$ such that $W\laweq X$  and 
$\E[Y|W]=W$.
Let $Z=Y-W$. It follows that
$\E[Y-W|W]=0$
and hence $\E[Z|Y-Z]=0$.  
 Therefore, $Z$ satisfies both $X \laweq Y-Z$ and \eqref{eq:new5}. 
\end{proof}

We make another remark regarding the   our results and
comonotonicity. Assume that $X$ and $Y$ in Theorem \ref{th:main} are comonotonic. By choosing $Z=Y-X$ and $W=X$,  condition \eqref{eq:new} becomes 
$$
\E[X|X\le x] \ge \E[Y|X\le x] \mbox{  for all relevant values of $x$.} 
$$
If $X$ is continuously distributed, then, using comonotonicity, this condition is
$$\int _0^p 
Q_X(t) \d t  \le \int _0^p  Q_Y (t) \d t  \mbox{~~~for all $p\in(0,1)$},
 $$
 which is a well-known equivalent condition for $X\ge _{\rm SSD} Y$ (see Lemma \ref{lem:1}).

We conclude the paper by discussing the interpretation of our representation results in relation to dependence concepts. 
As usual, $W$ is interpreted as the random wealth of a decision maker. 
In case $\E[Z]=0$, the condition \eqref{eq:new} 
  is   positive expectation dependence of $Z$ on $W$ (\cite{W87}); see \cite{LLW16} for its generalizations to higher order.
 Therefore, the implication from \eqref{eq:new} to 
 $W+Z\ge_{\rm cx} W$, which is Proposition \ref{pr:main} with $\E[Z]=0$,
yields the intuitive interpretation 
 that adding 
 a positively expectation dependent pertubation  
 increases the risk. Generally speaking, adding a positively dependent (in some vague sense) noise is risky; see \cite{DDGKV02} and \cite{PW15} for summaries on the intimate links between
 dependence concepts and stochastic orders.
 However, 
 if $\E[Z]\ne 0$, then \eqref{eq:new} no longer has an interpretation of positive dependence, as $\mathrm{cov}(W,Z)$ may be negative (see Example \ref{ex:2}).
Indeed, \eqref{eq:new} is strictly weaker than the combination of positive expectation dependence 
and $\E[Z]\le 0$.  
The resulting relation $W+Z\le_{\rm ssd} W$ now means that 
  an additive payoff with negative expected value 
in adverse scenarios of the random wealth $W$ (i.e., on events of the form $\{W\le x\}$), even if positively dependent,
makes the overall position less desirable for any decision makers who respect SSD, that is, those who prefer more wealth to less 
and are risk averse.

\subsection*{Acknowledgements}
The authors thank an Editor, two anonymous referees, David P.~Brown, and Felix Liebrich for helpful discussions and comments. 
Ruodu Wang acknowledges financial support from the Natural Sciences and Engineering Research Council of Canada (RGPIN-2024-03728) and Canada Research Chairs (CRC-2022-00141).

\end{document}